\documentclass[pra,nopacs,twocolumn,10pt]{revtex4}
\usepackage{amsmath}
\usepackage{graphicx,,epsfig,amsmath,latexsym,amssymb,verbatim,color}
\usepackage{dsfont}
\usepackage{MnSymbol}
\usepackage{enumerate}
 \usepackage{hyperref}
\hypersetup{colorlinks=true,citecolor=blue,linkcolor=blue,filecolor=blue,urlcolor=blue,breaklinks=true}

\usepackage{theorem}\usepackage{url}
\newtheorem{definition}{Definition}

\newtheorem{lemma}[definition]{Lemma}

\newtheorem{theorem}[definition]{Theorem}

\def\squareforqed{\hbox{\rlap{$\sqcap$}$\sqcup$}}
\def\qed{\ifmmode\squareforqed\else{\unskip\nobreak\hfil
\penalty50\hskip1em\null\nobreak\hfil\squareforqed
\parfillskip=0pt\finalhyphendemerits=0\endgraf}\fi}
\def\endenv{\ifmmode\;\else{\unskip\nobreak\hfil
\penalty50\hskip1em\null\nobreak\hfil\;
\parfillskip=0pt\finalhyphendemerits=0\endgraf}\fi}
\newenvironment{proof}{\noindent \textbf{{Proof~} }}{\qed}
\newenvironment{remark}{\noindent \textbf{{Remark~}}}{\qed}

\mathchardef\ordinarycolon\mathcode`\:
\mathcode`\:=\string"8000
\def\vcentcolon{\mathrel{\mathop\ordinarycolon}}
\begingroup \catcode`\:=\active
  \lowercase{\endgroup
  \let :\vcentcolon
  }

\newcommand{\nc}{\newcommand}
\nc{\rnc}{\renewcommand}
\nc{\beg}{\begin{equation}}
\nc{\eeq}{{\end{equation}}}
\nc{\beqa}{\begin{eqnarray}}
\nc{\eeqa}{\end{eqnarray}}
\nc{\lbar}[1]{\overline{#1}}
\nc{\bra}[1]{\langle#1|}
\nc{\ket}[1]{|#1\rangle}
\nc{\ketbra}[2]{|#1\rangle\!\langle#2|}
\nc{\braket}[2]{\langle#1|#2\rangle}

\nc{\proj}[1]{| #1\rangle\!\langle #1 |}
\nc{\avg}[1]{\langle#1\rangle}
\nc{\Rank}{\operatorname{Rank}}
\nc{\smfrac}[2]{\mbox{$\frac{#1}{#2}$}}
\nc{\tr}{\operatorname{Tr}}
\nc{\ox}{\otimes}
\nc{\dg}{\dagger}
\nc{\dn}{\downarrow}
\nc{\cA}{{\cal A}}
\nc{\cB}{{\cal B}}
\nc{\cC}{{\cal C}}
\nc{\cD}{{\cal D}}
\nc{\cE}{{\cal E}}
\nc{\cF}{{\cal F}}
\nc{\cG}{{\cal G}}
\nc{\cH}{{\cal H}}
\nc{\cI}{{\cal I}}
\nc{\cJ}{{\cal J}}
\nc{\cK}{{\cal K}}
\nc{\cL}{{\cal L}}
\nc{\cM}{{\cal M}}
\nc{\cN}{{\cal N}}
\nc{\cO}{{\cal O}}
\nc{\cP}{{\cal P}}
\nc{\cQ}{{\cal Q}}
\nc{\cR}{{\cal R}}
\nc{\cS}{{S}}
\nc{\cT}{{\cal T}}
\nc{\cX}{{\cal X}}
\nc{\cY}{{\cal Y}}
\nc{\cZ}{{\cal Z}}
\nc{\cW}{{\cal W}}
\nc{\csupp}{{\operatorname{csupp}}}
\nc{\qsupp}{{\operatorname{qsupp}}}
\nc{\var}{{\operatorname{var}}}
\nc{\rar}{\rightarrow}
\nc{\lrar}{\longrightarrow}
\nc{\polylog}{{\operatorname{polylog}}}
\nc{\1}{{\mathds{1}}}
\nc{\wt}{{\operatorname{wt}}}
\nc{\av}[1]{{\left\langle {#1} \right\rangle}}
\nc{\supp}{{\operatorname{supp}}}

\nc{\RR}{{{\mathbb R}}}
\nc{\CC}{{{\mathbb C}}}
\nc{\FF}{{{\mathbb F}}}
\nc{\NN}{{{\mathbb N}}}
\nc{\ZZ}{{{\mathbb Z}}}
\nc{\PP}{{{\mathbb P}}}
\nc{\QQ}{{{\mathbb Q}}}
\nc{\UU}{{{\mathbb U}}}
\nc{\EE}{{{\mathbb E}}}
\nc{\id}{{\operatorname{id}}}

\nc{\CHSH}{{\operatorname{CHSH}}}

\nc{\be}{\begin{equation}}
\nc{\ee}{{\end{equation}}}
\nc{\bea}{\begin{eqnarray}}
\nc{\eea}{\end{eqnarray}}
\nc{\<}{\langle}
\rnc{\>}{\rangle}
\nc{\Hom}[2]{\mbox{Hom}(\CC^{#1},\CC^{#2})}
\nc{\rU}{\mbox{U}}

\nc{\ob}[1]{#1}

\nc{\SEP}{{\text{SEP}}}
\nc{\NS}{{\text{NS}}}
\nc{\LOCC}{{\text{LOCC}}}
\nc{\PPT}{{\text{PPT}}}
\nc{\EXT}{{\text{EXT}}}
\nc{\Sym}{{\operatorname{Sym}}}

\nc{\ERLO}{{E_{\text{r,LO}}}}
\nc{\ERLOCC}{{E_{\text{r,LOCC}}}}
\nc{\ERPPT}{{E_{\text{r,PPT}}}}
\nc{\ERLOCCinfty}{{E^{\infty}_{\text{r,LOCC}}}}
\nc{\Aram}{{\operatorname{\sf A}}}

\begin{document}
\title{Indistinguishability of bipartite states by positive-partial-transpose operations in the many-copy scenario}
\author{Yinan Li$^{1}$}
\email{yinan.li@student.uts.edu.au}
\author{Xin Wang$^{1}$}
\email{xin.wang-8@student.uts.edu.au}

\author{Runyao Duan$^{1,2}$}
\email{runyao.duan@uts.edu.au}

\affiliation{$^1$Centre for Quantum Software and Information, Faculty of Engineering and Information Technology, University of Technology Sydney, NSW 2007, Australia}
\affiliation{$^2$UTS-AMSS Joint Research Laboratory for Quantum Computation and Quantum Information Processing, Academy of Mathematics and Systems Science, Chinese Academy of Sciences, Beijing 100190, China}

\begin{abstract} 
A bipartite subspace $S$ is called strongly positive-partial-transpose-unextendible (PPT-unextendible) if for every positive integer $k$, there is no PPT operator supporting on the orthogonal complement of $S^{\otimes k}$.  We show that a subspace is strongly PPT-unextendible if it contains a PPT-definite operator (a positive semidefinite operator whose partial transpose is positive definite). Based on these, we are able to propose a simple criterion for verifying whether a set of bipartite orthogonal quantum states is indistinguishable by PPT operations in the many copy scenario. Utilizing this criterion, we further point out that any entangled pure state and its orthogonal complement cannot be distinguished by PPT operations in the many copy scenario. On the other hand, we investigate that the minimum dimension of strongly PPT-unextendible subspaces in an $m\otimes n$ system is $m+n-1$, which involves a generalization of the result that non-positive-partial-transpose (NPT) subspaces can be as large as any entangled subspace [N. Johnston, Phys. Rev. A \textbf{87}: 064302 (2013)].
\end{abstract}
\maketitle
\section{Introduction}
One fascinating phenomenon of quantum mechanics is the quantum nonlocality, that is, there exist some global quantum operations on a composite system that cannot be implemented by the owners of the subsystems using local operations and classical communication (LOCC) only. A general strategy to study quantum nonlocality is to consider what kind of information processing tasks can be achieved by LOCC. Roughly speaking, if a certain task is accomplished with different optimal global and local efficiencies, then we can construct a class of quantum operations that cannot be realized by LOCC. There is no doubt that the discrimination of orthogonal quantum states is an effective one and received considerable attention in the past decades. See Refs. \cite{Bennett1999a, Walgate2000, Ghosh2001, Horodecki2003, Fan2004a, Ghosh2004, DeRinaldis2004, Nathanson2005, Watrous2005, Hayashi2006, Owari2006, Duan2007a, Feng2009, Yu2011, Bandyopadhyay2011b, Bandyopadhyay2011a, Yu2012, Lancien2013, Chitambar2014b} for a partial list. 

It is well-known that orthogonal quantum states can be perfectly distinguished if globe operations are permitted. And the set up of local distinguishability of quantum states is simple: two or more spatially separated observers share a composite quantum system prepared in one of many known mutually orthogonal quantum states. Their goal is to identify the unknown states by LOCC. When we only need to distinguish two orthogonal multipartite pure states, the perfect local discrimination can always be achieved perfectly \cite{Walgate2000}. Nevertheless, if we have more than two states to distinguish, they cannot be perfectly distinguished by LOCC if one or more states are entangled  \cite{Horodecki2003}. This phenomenon is percipient since entanglement has been shown to ensure difficulty in state discrimination \cite{Hayashi2006}. On the other hand, it has been shown that there exist sets of orthogonal product states that cannot be discriminated perfectly by LOCC \cite{Bennett1999a,Bennett1999b,Divincenzo2003}. 

However, the local indistinguishability might be overcome in the many copy scenario, which is the case that multiple copies of quantum states are provided. A simple example would be four $2\ox 2$ Bell states, which can be perfectly distinguished within two copies \cite{Ghosh2004}, while it is indistinguishable with a single copy \cite{Ghosh2004, Fan2004a, Nathanson2005}. In fact, it has been shown that $N$ orthogonal pure states can always be perfectly distinguished by LOCC when $N-1$ copies of the unknown state are provided \cite{Bandyopadhyay2011a}. This suggests that the available copies of the unknown state play a crucial role in local distinguishability. However, it turns out that there exist two quantum states, one of which is necessarily mixed, are locally indistinguishable, in the many copy scenario \cite{Bandyopadhyay2011a}. Thus, the local indistinguishability of orthogonal quantum states is more robust in mixed states for it persists even in the domain of multiple copies, whereas in the case of pure states it does not. 

Proving local indistinguishability is hard even in the bipartite case, since our knowledge about LOCC is limited.  To circumvent the difficulty, one approach is to show the indistinguishability by operations completely preserving positivity of partial transpose (denote it shortly by PPT operations), and local indistinguishability automatically follows since the set of all LOCC operations will also preserving positivity of partial transpose. The advantage of this approach is that the set of PPT operations enjoys a tractable mathematical structure. It has been shown that the bipartite maximally entangled state and the normalized projection onto its orthogonal complement cannot be distinguished by PPT operations in the many copy scenario \cite{Yu2014a}. On the other hand, the notion of PPT plays a significant role in quantum information theory. It has been used to provide some convenient criterion for the separability of quantum states \cite{Peres1996,Horodecki1996,Horodecki1998}, and study the problem of entanglement distillation, pure state transformation and communication over quantum channels (e.g. \cite{Rains2001,Wang2016c,Wang2016,Ishizaka2004,Matthews2008,Leung2015c,Wang2016g}).  

In this paper, we contribute a simple criterion for verifying indistinguishability of bipartite quantum states by PPT operations, in the many copy scenario. This criterion is based on the observation that, if the support of some state is strongly PPT-unextendible, then any set of orthogonal states containing it cannot be distinguished by PPT operations in the many copy scenario. Here we say a subspace $S$ of some bipartite Hilbert space \textit{Strongly PPT-unextendible} if for any positive integer $k$, there is no PPT state whose support is a subspace of the orthogonal complement of $S^{\ox k}$. And this concept is a natural generalization of the concept of \textit{strongly unextendible subspace}, introduced in Ref. \cite{Cubitt2011a}.  To observe the strongly PPT-unextendibility, one witness is the so-called PPT-definite operator, which is a positive semidefinite operator whose partial transpose is positive definite. In fact, the existence of PPT-definite operators can be observed efficiently by semidefinite programming (SDP) \cite{Vandenberghe1996}, which is a powerful tool in quantum information theory with many applications (e.g., \cite{Piani2015,Jain2011a,Wang2016a,Duan2016,Skrzypczyk2014,Wang2016d,Berta2015b,Doherty2002a}). As an application, we show that any entangled pure states and the normalized projector onto their orthogonal complement cannot be distinguished by PPT operations in the many copy scenario, which is a far-reaching extension of one of the main results in Ref. \cite{Yu2014a}. Meanwhile, we show that the minimum dimension of strongly PPT-unextendible subspaces in a $m\otimes n$ system is $m+n-1$, which extends the result on the minimum dimension of the PPT-unextendible subspace in Ref. \cite{Johnston2013}. 

\section{Preliminaries}
We review some notations and definitions. In the following, we will use symbols such as $\cA$ (or $\cA'$) and $\cB$ (or $\cB'$) to denote (finite-dimensional) Hilbert spaces associated with Alice and Bob, respectively. The dimension of $\cA$ and $\cB$ are denoted by $d_\cA$ and $d_\cB$. We say two subspaces $S_1$ and $S_2$ of some Hilbert space are orthogonal, denoted by $S_1\perp S_2$, if for any $\ket{\psi_1}\in S_1$ and $\ket{\psi_2}\in S_2$, $\braket{\psi_1}{\psi_2}=0$. The orthogonal complement of a subspace $S$ is denoted by $S^{\perp}=\{\ket{\psi}:\braket{\psi}{\phi}=0~\forall\ket{\phi}\in S\}$. The space of all linear operators over $\cA$ is denoted by $\cL(\cA)$. For convenience, we use $\lambda_{\max}(X)$ and $\lambda_{\min}(X)$ to denote the maximum eigenvalue and the minimum eigenvalue of some operator $X\in\cL(\cA)$. A quantum state is characterized by its density operator $\rho\in\cL(\cA)$, which is a positive semidefinite operator with trace unity.  The support of $\rho$, denoted by $\supp(\rho)$, is defined to be the space spanned by the eigenvectors of $\rho$ with positive eigenvalues. We say a positive semidefinite operator $X$ is supporting on some subspace $S$ of $\cA$, if $\supp(X)$ is a subspace of $S$. 

A bipartite positive semidefinite operator $E_{AB} \in \cL(\cA\ox \cB)$ is
said to be Positive-Partial-Transpose (PPT) if $E_{AB}^{T_{B}}$ is positive semidefinite, where the action of partial transpose (with respect to $B$) is defined as $(\ket{i_A}\bra{k_A}\otimes\ket{j_B}\bra{l_B})^{T_{B}}=\ket{i_A}\bra{k_A}\otimes\ket{l_B}\bra{j_B}$. Moreover, a PPT operator $E_{AB}\in\cL(\cA\ox \cB)$ is said to be \emph{PPT-definite}, if $E_{AB}^{T_{B}}$ is positive definite.

In this paper, the PPT operations used for distinguishing a set of $n$ orthogonal quantum states $\{\rho_1,\dots,\rho_n\}$ can be defined as a $n$-tuple of operators, $(M_k)_{k=1,\dots,n}$, where $M_k\in \cL(\cA\otimes\cB)$ is PPT for $k=1,\dots,n$ and $\sum_{k=1}^nM_k=\1_{\cA\ox\cB}$. Then $\{\rho_1,\dots,\rho_n\}$ is said to be
\begin{enumerate}[(i)]
\item \emph{perfectly distinguishable by PPT operations}, if there exist $(M_k)_{k=1,\dots,n}$, where $M_k\in \cL(\cA\otimes\cB)$ is PPT for $k=1,\dots,n$ and $\sum_{k=1}^nM_k=\1_{\cA\ox\cB}$, such that 
$\tr(M_i\rho_j)=\delta_{ij},$
for any $1\le i,j \le n$;
\item \emph{unambiguously distinguishable by PPT operations}, if there exist $(M_k)_{k=1,\dots,n}$, where $M_k\in \cL(\cA\otimes\cB)$ is PPT for $k=1,\dots,n$ and $\sum_{k=1}^nM_k=\1_{\cA\ox\cB}$, such that 
$\tr(M_i\rho_j)=p_i\delta_{ij}, p_i>0$
for any $1\le i,j \le n$;
\item \emph{indistinguishable by PPT operations}, if it is not unambiguously distinguishable by PPT operation.
\end{enumerate}

These definitions can be naturally generalized when multiple copies are provided. In addition, we would say a set of orthogonal quantum states $\{\rho_1, ..., \rho_n\}$ is indistinguishable by PPT operations in the many copy scenario, if for any positive integer $k$, $\{\rho_1^{\ox k}, ..., \rho_n^{\ox k}\}$ is indistinguishable by PPT operations.

In the end, we say a bipartite subspace $S$ of $\cA\ox\cB$ is said to be PPT-extendible, if there exists a PPT operator $\sigma\in\cL(\cA\ox \cB)$, such that $\sigma$ is supporting on the orthogonal complement of $S$. $S$ is said to be PPT-unextendible if it is not PPT-extendible, and to be strongly PPT-unextendible if for any positive integer $k$, $S^{\ox k}$ is not PPT-extendible. 

\section{Main results}
\subsection{Indistinguishability by PPT operations}
The main result of this paper is a sufficient criterion for verifying the PPT indistinguishability of orthogonal quantum states:
\begin{theorem}\label{P PPT-definite}
For a set of orthogonal bipartite quantum states $\{\rho_1, ..., \rho_n\}$,
if there is a PPT-definite operator supporting on the support of some $\rho_k$, then $\{\rho_1, ..., \rho_n\}$ is indistinguishable by PPT operations in the many copy scenario.
\end{theorem}

To prove this theorem, we first show the following lemma.
\begin{lemma}\label{Th 1}
For a set of orthogonal quantum states $\{\rho_1, ..., \rho_n\}$, if there exists $k\in\{1,\dots,n\}$ such that $\supp(\rho_k)$ is strongly PPT-unextendible, then $\cS$ is  indistinguishable by PPT operations in the many copy scenario.
\end{lemma}
\begin{proof}
Without loss of generality, we assume $\supp(\rho_1)$ is strongly PPT unextendible. If there exists some positive integer $m$ such that $\{\rho_1^{\ox m},\dots,\rho_n^{\ox m}\}$ can be unambiguously distinguished by PPT operations, there exists a tuple of PPT operators $(M_k)_{k=1,\dots,n}$ such that
$$\tr(M_kP_j^{\ox m})=p_k\delta_{kj},~~ p_k>0,~~j=1,\dots,m,$$
where $P_j$ is the projection onto $\supp(\rho_j^{\ox m})$ for any $1\le j\le n$ and $p_k>0$ for any $1\le k\le n$. 
Notice that $M_1$ will support on the orthogonal complement of $\supp(\rho_1)$, which is a contradiction.
\end{proof}

Now, it is sufficient to show that the property of strongly PPT-unextendible can be observed by PPT-definite operator:

\begin{lemma}\label{perp PPT}
Given a bipartite subspace $S$ of $\cA\ox \cB$, if there is a PPT-definite operator $\sigma$ supporting on $S$, then $S$ is strongly PPT unextendible.
\end{lemma}
\begin{proof}
Assume that there exists $k$ such that $S^{\ox k}$ is PPT-extendible and $\rho\in\cL(\cA^{\ox k}\ox \cB^{\ox k})$ is the PPT operator supporting on the orthogonal complement of $S^{\ox k}$. We will show
\begin{equation}\label{perp no PPT}
\tr (\sigma^{\ox k}\rho)=\tr((\sigma^{\ox k})^{T_B}\rho^{T_B})>0,
\end{equation}
which is a contradiction since $\supp(\sigma^{\ox k})$ is still a subspace of $\supp(S^{\ox k})$. To see this, we show $\sigma^{\ox k}$ is still PPT-definite. 
If we let $P$ to be the projection onto $\supp(\sigma)$, then $\sigma$ is PPT-definite if and only if $T(\sigma)>0$, where
\begin{equation}\label{prime T}
\begin{split}
T(\sigma)= \max & \ t, \\
\phantom{T(\sigma) }\text{ s.t. }  &  0\le R\le P, \ R^{T_B}\ge t\1.
\end{split}
\end{equation}
Actually, the function $T(\cdot)$ is the super-multiplicative,  i.e. 
$$T(\sigma_1\ox\sigma_2)\ge T(\sigma_1)T(\sigma_2).$$ 
To prove this, we can assume that the optimal solutions to SDP (\ref{prime T}) of $T(\sigma_1)$ and $T(\sigma_2)$ are $\{R_1, t_1\}$ and $\{R_2, t_2\}$, respectively. It is clear that $0\le R_1\ox R_2\le P_1\ox P_2$ and $ R_1^{T_{B_1}}\ox R_2^{T_{B_2}}\ge t_1t_2\1$. Then $\{R_1\ox R_2, t_1t_2 \}$ is a feasible solution to   SDP (\ref{prime T}) of $T(\sigma_1\ox\sigma_2)$, which means that $T(\sigma_1\ox\sigma_2)\ge t_1t_2>0$.
It follows immediately that $\sigma^{\ox n}$ is also PPT-definite, since $T(\sigma^{\ox n})\ge T(\sigma)^n>0$.
\end{proof}

The proof of theorem \ref{P PPT-definite} is then straightforward. If there is a PPT-definite operator supporting on the support of some $\rho_k$, by lemma \ref{perp PPT}, $\supp(\rho_k)$ is strongly PPT-unextendible; Then by lemma \ref{Th 1}, $\{\rho_1,\dots,\rho_n\}$ is PPT indistinguishable in the many copy scenario. \qed

\subsection{Examples of sets of quantum states which are PPT indistinguishable in the many copy scenario}
To see the power of our theorem, we apply it to extend the result Ref. \cite{Yu2014a}, which showed that the bipartite maximally entangled state and the normalized projection onto its orthogonal complement are PPT indistinguishable in the many copy scenario. In fact, our proof illustrates that the restriction to maximally entangled state can be removed.
\begin{theorem}\label{pure state}
Given any \textbf{entangled} state $\ket{\phi}\in\cA\otimes\cB$ with $d_{\cA}=d_{\cB}=d$, let $\rho=\frac{1}{d^2-1}(\1_{\cA\ox \cB}-\proj{\phi})$ be the normalized projection onto its orthogonal complement. Then $\proj{\phi}$ and $\rho$ are PPT indistinguishable in the many copy scenario.
\end{theorem}
\begin{proof}
It is easy to see that $\lambda_{\min}((\1-\proj{\phi})^{T_B})>0$ is equivalent to  $\lambda_{\max}(\proj \phi^{T_B})<1$.
Suppose that the Schmidt rank of $\ket{\phi}$ is $m$ ($>1$), and the Schmidt decomposition of $\ket{\phi}$ is $\ket{\phi}=\sum_{i=1}^{m}\lambda_i\ket{ii}$ with $\lambda_1^2\ge...\ge\lambda_m^2$ and $\sum_{i=1}^m\lambda_i^2=1$. The partial transposition (with respect to $\cB$) of $\proj{\phi}$ is:
\begin{equation}
\begin{split}
\proj{\phi}^{T_B}&=\sum_{i=1}^m\lambda_i^2\proj{ii}+\sum_{i\ne j}\lambda_i\lambda_j\ket{ji}\bra{ij}\\
&=\sum_{i=1}^m\lambda_i^2\proj{ii}+\sum_{i>j}\frac{\lambda_i\lambda_j}{2}[(\ket{ij}+\ket{ji})(\bra{ij}+\bra{ji})\\&+(\ket{ij}-\ket{ji})(\bra{ij}-\bra{ji})].
\end{split}
\end{equation}
This shows that $\lambda_{\max}(\proj \phi^{T_B})=\lambda_1^2<1$.
Therefore, $(\1-\proj{\phi})^{T_B}$ is positive definite and the result follows directly from Theorem \ref{P PPT-definite}.
\end{proof}
\medskip

On the other hand, it is also interesting to construct sets of PPT indistinguishable orthogonal quantum states in the many copy scenario without invoking technique of unextendible product bases (UPBs) \cite{Bennett1999a}.  Utilizing theorem \ref{P PPT-definite}, we exhibit one simple example as follows:
\medskip
\newline
\textbf{Example 1} Let $d_\cA=d_\cB=d$. Choose a set of orthogonal basis $\{\ket{\phi_1},\dots,\ket{\phi_{d^2}}\}$ of $\cA\ox \cB$, such that $\ket{\phi_i}$ is maximally entangled for any $i=1,\dots, d^2$. For any positive integer $d^2-d+1\le m\le d^2$ and any $k=2,\dots,d^2-m+1$, we can construct $\{\rho_i:i=1,\dots, k\}$ which is PPT indistinguishable in the many copy scenario.
\medskip

\begin{proof}
We first show that the projection $P$ onto the subspace ${\rm span}\{\ket{\phi_1},\dots,\ket{\phi_m}\}$ is PPT-definite.
Notice that $P=\1_{d^2}-\sum_{i=m+1}^{d^2}\proj{\phi_i}$, then
\begin{align}
\lambda_{\min}(P^{T_B})&=1-\lambda_{\max}(\sum_{i=m+1}^{d^2}\proj{\psi_i}^{T_B})\\
&\ge 1-\sum_{i=m+1}^{d^2}\lambda_{\max}(\proj{\psi_i}^{T_B})\\
&\ge 1-\frac{d^2-m}{d}\ge\frac{1}{d}>0, \label{MES T eig}
\end{align}
where Ineq. (\ref{MES T eig}) uses the fact that $\lambda_{\max}(\proj{\psi_i}^{T_B})=1/d$.

Then $\{\rho_1,\dots,\rho_k\}$ can be chosen as:
\begin{align*}
\rho_1=&\frac{1}{m}\sum_{i=1}^{ m}\proj{\phi_i},\\
\rho_2=&\proj{\phi_{m+1}},\\
&...\\
\rho_k=&\proj{\phi_{m+k-1}},
\end{align*}
where $k\le d^2-m+1$.
\end{proof}
\begin{remark}
For a general set of pure states $\{\ket{\psi_1},\dots,\ket{\psi_{m}}\}$, we suppose $\lambda_i$ is the largest Schmidt coefficient of $\ket{\psi_i}$. We can use similar technique to show that if $\sum_{i=1}^{m} \lambda_i<1$, then any set of states $\{\rho_1,...,\rho_k\}$ with $\rho_1=\frac{1}{m}\sum_{i=1}^n\proj{\psi_i}$ is PPT indistinguishable in the many copy scenario. 
\end{remark}

\subsection{Minimum dimension of strongly PPT-unextendible subspace}
It is of great interest to connect PPT distinguishability with other concepts in quantum information theory, and there is no doubt that lemma \ref{Th 1} provides one.  The notion of strongly PPT unextendible subspace is a natural generalization of strongly unextendible subspace, while the latter has been shown widely useful in quantum information theory. For instance, it has been shown in Ref. \cite{Cubitt2008a} that the minimum dimension of such a subspace is $d_\cA+d_\cB-1$, and this result has been applied to show the superactivation of the asymptotic zero-error classical capacity of a quantum channel \cite{Cubitt2011a,Duan2010}. Since there is also no product state in the orthogonal complement of PPT-unextendible subspaces, the minimum dimension of the strongly PPT-unextendible subspace is at least $d_\cA+d_\cB-1$. In Ref. \cite{Johnston2013}, a PPT-unextendible subspace of dimension $d_\cA+d_\cB-1$ has been explicitly constructed  \cite{Johnston2013}. To be specific, this subspace is the orthogonal complement of 
$$S={\rm span}\{\ket{j}\ket{k+1}-\ket{j+1}\ket{k}:~~0\leq j\leq d_{\cA}-2,~~0\leq k\leq d_{\cB}-2\}.$$ 
Using Lemma \ref{perp PPT}, we show that the subspace $\cS^\perp$ is actually strongly PPT-unextendible, which illustrated that the minimum dimension of strongly PPT-unextendible subspaces in $\cA\ox \cB$ is also $d_\cA+d_\cB-1$.

\begin{theorem}\label{minimum size}
Let $d_\cA=m$, $d_\cB=n$ satisfying $2\leq m\leq n$, and $S$ is defined as above. Then $S^\perp$ is strongly PPT-unextendible.
\end{theorem}
\begin{proof}
Denote $S^\perp$ by $S_{mn}$ with respect to the dimension of $\cA$ and $\cB$. $S_{mn}$ can be written in the following form:
\begin{equation*}
\begin{split}
S_{mn}=&{\rm span}\{\ket{\psi_s}=\sum_{j=0}^{m-1-s}\ket{j}\ket{m-1-s-j}:~s=0,\dots,m-1; \\
&\ket{\phi_t}=\sum_{j=t-m+1}^{\min\{n-1,t\}}\ket{t-j}\ket{j}:~t=m,\dots,m+n-2\}.
\end{split}
\end{equation*}

We claim that there exists positive real numbers $x_0,x_1,\dots,x_{m-1},y_{m},\dots,y_{m+n-2}$ such that 
\begin{equation}\label{rho1}
\rho_{mn}=\sum_{s=0}^{m-1}x_{m-1-s}\proj{\psi_s}+\sum_{t=m}^{m+n-2}y_{t}\proj{\phi_t}
\end{equation}
is PPT-definite. 

Notice that
\begin{equation}
\begin{split}
\proj{\psi_s}^{T_B}&=\sum_{j_1,j_2=0}^{n-1-s}\ket{j_1}\bra{j_2}\otimes\ket{m-1-s-j_2}\bra{m-1-s-j_1},\\
\proj{\phi_t}^{T_B}&=\sum_{j_1,j_2=t-m+1}^{\min\{n-1,t\}}\ket{t-j_1}\bra{t-j_2}\otimes\ket{j_2}\bra{j_1}.
\end{split}
\end{equation}
We consider the matrix form of $\rho_{mn}^{T_B}$ under the computational basis. Divide $\{\ket{jk}:~j=0,\dots,m-1,~k=0,\dots,n-1\}$ into the following families:
\begin{equation}
\begin{split}
&\cP_a=\{\ket{m-1-a+t}\ket{t}:~0\leq t\leq a\}~a=0,\dots,m-1;\\
&\cQ_b=\{\ket{r}\ket{r+b}:~0\leq r\leq \min\{n-1-b,m-1\}\}~b=1,\dots,n-1.
\end{split}
\end{equation}
The submatrices spanned by $\cP_a$ and $\cQ_b$ are denoted by $P_a$ and $Q_b$. More precisely, $P_{a}$ and $Q_{b}$ have the following form:
\begin{equation}\label{matrices}
\begin{split}
P_a&=\begin{pmatrix}
x_{a} & x_{a-1} &\cdots & x_0\\
x_{a-1} & \reflectbox{$\ddots$} & \reflectbox{$\ddots$} & y_m\\
\reflectbox{$\ddots$}    & \reflectbox{$\ddots$} & \reflectbox{$\ddots$} & \vdots\\
x_0 &y_m &\cdots  &y_{m+a-1}\\
\end{pmatrix}_{(a+1)\times (a+1)}~{ 0\leq a\leq m-1},\\
Q_b&=\begin{pmatrix}
x_{m-1-b} &\cdots& x_{0} &y_m & \cdots & y_{m+b-1}\\
\vdots & \reflectbox{$\ddots$} & \reflectbox{$\ddots$} & \reflectbox{$\ddots$} &\reflectbox{$\ddots$} &\vdots\\
x_0    & \reflectbox{$\ddots$} & \reflectbox{$\ddots$} & \reflectbox{$\ddots$} &\reflectbox{$\ddots$} &\vdots\\
y_m    & \reflectbox{$\ddots$} & \reflectbox{$\ddots$} &\reflectbox{$\ddots$}  &\reflectbox{$\ddots$} &\vdots\\
\vdots   & \reflectbox{$\ddots$} & \reflectbox{$\ddots$} &\reflectbox{$\ddots$}  &\reflectbox{$\ddots$} &\vdots\\
y_{m+b-1} &\cdots& \cdots &\cdots & \cdots & y_{2m+b-2}\\
\end{pmatrix}_{(m)\times (m)}~{1\leq b\leq n-m},\\
Q_b&=\begin{pmatrix}
x_{m-1-b} &\cdots& x_{0} &y_m & \cdots & y_{n-1}\\
\vdots & \reflectbox{$\ddots$} & \reflectbox{$\ddots$} & \reflectbox{$\ddots$} &\reflectbox{$\ddots$} &\vdots\\
x_0    & \reflectbox{$\ddots$} & \reflectbox{$\ddots$} & \reflectbox{$\ddots$} &\reflectbox{$\ddots$} &\vdots\\
y_m    & \reflectbox{$\ddots$} & \reflectbox{$\ddots$} &\reflectbox{$\ddots$}  &\reflectbox{$\ddots$} &\vdots\\
\vdots   & \reflectbox{$\ddots$} & \reflectbox{$\ddots$} &\reflectbox{$\ddots$}  &\reflectbox{$\ddots$} &\vdots\\
y_{n-1} &\cdots& \cdots &\cdots & \cdots & y_{2n-b-2}\\
\end{pmatrix}_{(n-b)\times (n-b)}~{b > n-m}.
\end{split}
\end{equation}
Then we have 
$$\rho_{mn}^{T_B}=(\oplus_{a=0}^{m-1} P_a)\oplus(\oplus_{b=1}^{n-1}Q_b).$$

To make sure that $\rho_{mn}^{T_B}$ is positive definite, it is equivalent to make sure that $P_a$ and $Q_b$ are positive definite for $a=0,\dots,m-1$ and $b=1\dots, n-1$.

The case $m=n=2$ is easy to verify. Notice that $S_{22}={\rm span}\{\ket{00},\ket{01}+\ket{10},\ket{11}\}$ and we can easily choose a  operator  $\rho=2(\proj{00}+\proj{11})+(\ket{01}+\ket{10})(\bra{01}+\bra{10})$ which is PPT-definite.

We first consider $m=n$ and prove it by mathematical induction. Specifically, if $\rho_{mm}^{T_B}$ is positive definite, then we can construct $\rho_{m+1m+1}^{T_B}$ which is also positive definite.  Notice that when $m=n$, we can assume $x_a=y_{m-1+a}$, then we have $P_a=Q_{m-1-a}$ for $a=0,\dots,m-2$ and $\rho_{mm}^{T_B}=P_{m-1}\oplus [\oplus_{k=0}^{m-2}(P_k\oplus Q_{m-1-k})]$. Since $\rho_{mm}^{T_B}$ is positive definite, there exist positive $x_0,\dots,x_{m-1}$ such that 
$$P_k=\begin{pmatrix}
x_{k} & x_{k-1} &\cdots & x_0\\
x_{k-1} & \reflectbox{$\ddots$} & \reflectbox{$\ddots$} & x_1\\
\vdots    & \reflectbox{$\ddots$} & \reflectbox{$\ddots$} & \vdots\\
x_0 &x_1 &\cdots  &x_{k}\\
\end{pmatrix}_{(k+1)\times (k+1)}$$
is positive definite for $k=0,\dots,m-1$. For $\rho_{(m+1)(m+1)}^{T_B}$, we want to find $x_0',\dots,x_{m}'$ such that 
$$P_k'=\begin{pmatrix}
x_{k}' & x_{k-1}' &\cdots & x_0'\\
x_{k-1}' & \reflectbox{$\ddots$} & \reflectbox{$\ddots$} & x_1'\\
\vdots    & \reflectbox{$\ddots$} & \reflectbox{$\ddots$} & \vdots\\
x_0' &x_1' &\cdots  &x_{k}'\\
\end{pmatrix}_{(k+1)\times (k+1)}>0$$
for $k=0,\dots,m$. Let $x_k'=x_k$ for $k=0,\dots,m-1$, then $P_k'>0$ can be guaranteed. We only need to find a positive $x_m'$ such that 
$$P_m'=\begin{pmatrix}
x_{m}' & x_{m-1} &\cdots & x_0\\
x_{m-1} & \reflectbox{$\ddots$} & \reflectbox{$\ddots$} & x_1\\
\vdots    & \reflectbox{$\ddots$} & \reflectbox{$\ddots$} & \vdots\\
x_0 &x_1 &\cdots  &x_{m}'\\
\end{pmatrix}_{(m+1)\times (m+1)}$$
is positive definite. Notice that we only need to show the leading principal minors of $P'_m$ are all positive definite. Thus we have $m-1$ linear constraints and a quadratic constraint on variable $x_{m}'$. For all linear constraints, the coefficient of $x_{m}'$ are positive, which can be easily derived since $P'_k$ are positive definite for $k=1,\dots,m-1$. Moreover, the coefficient of $(x_m')^2$ in the quadratic constraint is also positive since the following matrix
$$\begin{pmatrix}
x_{m-2} & \cdots & x_1 & x_0\\
\vdots    & \reflectbox{$\ddots$} & x_0 & x_1\\
x_1    & \reflectbox{$\ddots$} & \reflectbox{$\ddots$} & \vdots\\
x_0 &x_1 &\cdots  &x_{m-2}\\
\end{pmatrix}$$
is also positive definite. It is then straightforward to see that feasible solutions are always exist. Thus we can choose one to make sure $\rho_{(m+1)(m+1)}^{T_B}$ is positive definite.

Finally, we extend to the case $m \ne n$ by similar technique. Assume we have already known $x_0',\dots,x_{m-1}'$ such that $\rho_{mm}$ is PPT-definite. For $\rho_{mn}$ where $n>m$, we are going to find $x_0,\dots,x_{m-1}$ and $y_m,\dots, y_{m+n-2}$ such that $\rho_{mn}$ is PPT-definite too. Let $x_i=x_i'$ for $i=0,\dots,m-1$ and $y_{m+j}=x_{j+1}$ for $j=0,\dots, m-2$.  These guarantee that $P_a$ in Eq. \ref{matrices} is positive definite for $a=0,\dots, m-1$ . Then we consider $Q_b$, where $b=1\dots, n-m$. When $b=1$, we only need to determine $y_{2m-1}$ such that $y_{2m-1}>0$ and $Q_1$ is positive definite. This can be done by only guarantee the determinant of $Q_b$ is always positive, which is a linear constraint for $y_{2m-1}$. This is true since the first $m-1$ leading principle minors of $Q_1$ are leading principle minors of $P_{m-2}$, which is automatically positive. Then we can determine $y_{2m-2},\dots, y_{m+n-2}$ with the same technique. Finally, we show that $Q_b$ is positive definite for $b=n-m+1,\dots,n-1$. In fact, this is the $m-1$th leading principle minors of some $Q_{n-m-2}$, which is positive definite. This concludes our proof.
\end{proof}
\medskip

\section{Conclusions and Discussions}
In summary, we study the indistinguishability of bipartite quantum states by PPT operations in the many copy scenario. By introducing the concept of strongly PPT-unextendible subspace, we show that such subspace plays crucial role in determining PPT indistinguishability in the many copy scenario. Driven by that, we show that PPT-definite operators can be served as a witness for strongly PPT-unextendibility. And this witness can be formalized as a semidefinite program, which can be checked efficiently. 

We then apply our result to demonstrate that any entangled pure state and the normalized projector onto its orthogonal complement is PPT indistinguishable in the many copy scenario. This provides a simpler and more general proof than that in Ref. \cite{Yu2014a}. On the other hand, we apply our results to show that the minimum dimension of strongly PPT-unextendible subspaces in an $m\ox n$ quantum system is $m+n-1$. This coincides with the minimum dimension of unextendible subspace \cite{Cubitt2008a} and involves an extension of the result that NPT subspaces can be as large as any entangled subspace in Ref. \cite{Johnston2013}.

The authors would like to thank  Yuan Feng, Andreas Winter and Dong Yang for helpful discussions. This work was partly supported by the Australian Research Council under Grant Nos. DP120103776 and FT120100449.
%

\end{document}